\documentclass[a4paper,11pt]{article}

\usepackage[twoside,a4paper,margin=1in]{geometry}
\usepackage[utf8]{inputenc}
\usepackage{graphicx,url,amssymb,amsmath,amsthm}
\usepackage{thm-restate}
\usepackage{xcolor}
\usepackage{todonotes}

\theoremstyle{plain}
\newtheorem{theorem}{Theorem}[section]
\newtheorem{proposition}[theorem]{Proposition}

\newtheorem{claim}[theorem]{Claim}
\newtheorem{lemma}[theorem]{Lemma}
\newtheorem{remark}[theorem]{Remark}

\usepackage{hyperref}

\author{Vincent Delecroix\thanks{Univ. Bordeaux, CNRS, Bordeaux INP, LaBRI, UMR 5800, F-33400 Talence, France, \href{mailto:vincent.delecroix@u-bordeaux.fr}{vincent.delecroix@u-bordeaux.fr}}\and Oscar Fontaine\thanks{Univ. Bordeaux, CNRS, Bordeaux INP, LaBRI, UMR 5800, F-33400 Talence, France, \href{mailto:oscar.fontaine@u-bordeaux.fr}{oscar.fontaine@u-bordeaux.fr}} \and Arnaud de Mesmay\thanks{Univ Gustave Eiffel, CNRS, LIGM, F-77454 Marne-la-Vallée, France, \href{mailto:arnaud.de-mesmay@univ-eiffel.fr}{arnaud.de-mesmay@univ-eiffel.fr}}}

\renewcommand{\ge}{\geqslant}
\renewcommand{\le}{\leqslant}
\renewcommand{\geq}{\geqslant}
\renewcommand{\leq}{\leqslant}
\renewcommand{\epsilon}{\varepsilon}

\newcommand{\tw}{\operatorname{tw}}
\newcommand{\rad}{\operatorname{rad}}
\newcommand{\crossnum}{\operatorname{cr}}

\definecolor{definecolor}{rgb}{0,0.1,0.55}
\def\define#1{\textbf{\textcolor{definecolor}{#1}}}

\everymath{\displaystyle}

\title{On the size of \texorpdfstring{$k$}{k}-irreducible triangulations\thanks{This work was funded by the ANR-SNF project SUGAR (ANR-25-CE40-0416, SNF 200021E\_238147). Work of the first author was funded by ANR MOST (ANR-23-CE40-0020) and ANR CarteEtPlus (ANR-23-CE48-0018).}}

\begin{document}

\maketitle

\begin{abstract}
    A triangulation of a surface is $k$-irreducible if every non-contractible curve has length at least $k$ and any edge contraction breaks this property. Equivalently, every edge belongs to a non-contractible curve of length $k$ and there are no shorter non-contractible curves. We prove that a $k$-irreducible triangulation of an orientable surface of genus $g$ has $O(k^2g)$ triangles, which is optimal. This is an improvement over the previous best bound $k^{O(k)} g^2$ of Gao, Richter and Seymour [\emph{Journal of Combinatorial Theory, Series B, 1996}].
\end{abstract}

\section{Introduction}
The main objects of study in this article are \define{(surface) triangulations}. There are many inequivalent definitions of those, and for us they will be simplicial complexes homeomorphic to surfaces. Rephrased in a graph-theoretical language, this means that a triangulation is a graph embedded in a surface where all the faces are topological disks and have degree three, and the graph is \define{simple}: multiple edges and loops are forbidden. 

In 1922, Steinitz~\cite{steinitz2013vorlesungen} proved that all triangulations of the $2$-dimensional sphere could be obtained from the tetrahedron by \emph{splitting vertices}. Nowadays, it is more common to consider the inverse operation of \define{contracting an edge}, which consists in identifying the endpoints of an edge and collapsing the two adjacent triangles. Such a contraction must preserve the properties of triangulations, and in particular must not induce multiple edges and loops. Thus, Steinitz proved that the tetrahedron is the unique \define{irreducible triangulation} of the sphere: every triangulation of the sphere can be reduced to the tetrahedron by iteratively contracting edges, and the tetrahedron cannot be contracted further.

Barnette initiated the study of the analogous problem on the projective plane~\cite{barnette1982generating}, which led to the classification of irreducible triangulations of surfaces up to Euler genus $4$ by Sulanke~\cite{sulanke2006generating} by means of the computer program \texttt{surftri}. It is now known that for every surface of Euler genus $g$, the number of irreducible triangulations is finite and that they have at most $13g-4$ vertices, see Joret and Wood~\cite{joret2010irreducible}. The latter improves previous bounds by Barnette and Edelson~\cite{BarnetteEdelson1989} and by Nakamoto and Ota~\cite{NakamotoOta1995}. This bound has been extended to surfaces with boundaries by Boulch, Colin de Verdi{\`e}re and Nakamoto~\cite{BoulchColinDeVerdiereNakamoto2013}.

A closed curve on a surface is \define{non-contractible} if it is not homotopic (intuitively, if it cannot be deformed continuously) to a point. Any closed walk on the edges of a triangulation determines a closed curve on the underlying surface. The \define{edge-width} of a triangulation is the length of the shortest (by the number of edges) non-contractible closed walk. A triangulation is \define{$k$-irreducible} if it has edge-width at least $k$ and it is irreducible for this property: any edge contraction reduces this edge-width. Equivalently, every edge belongs to a non-contractible curve of length $k$ and there are no shorter non-contractible curves. We refer to the PhD thesis of Melzer for a proof of this equivalence~\cite[Lemma~3.1]{melzer2019k} and for a thorough introduction to $k$-irreducible triangulations, in particular the different motivations for which they have been introduced.

Since loops and multiple edges are forbidden, triangulations of surfaces of positive Euler genus have edge-width at least $3$. From the definitions, it is easy to see that irreducible triangulations and $3$-irreducible triangulations match. Therefore, $k$-irreducible triangulations generalize irreducible triangulations, and one is naturally led to wonder about their combinatorial properties. 

This problem was first studied by Malni\v{c} and Nedela~\cite{malnic1995k} who proved that the number of $k$-irreducible triangulations of a surface is finite (they also observed that this follows from a variant of Robertson-Seymour theory), see also Juvan, Malni\v{c} and Mohar~\cite{juvan1996systems}. The current best known bound, due to Gao, Richter and Seymour~\cite{gao1996irreducible} is that a $k$-irreducible triangulation of a surface of Euler genus $g$ has $k^{O(k)} g^2$ edges. In that paper, the authors say that they expect the correct bound to depend linearly on $g$ instead of quadratically. In his aforementioned PhD thesis, Melzer~\cite[p.~18]{melzer2019k} states ``Our expectation would be something like $O(k^2 g)$, but we lack any evidence at all to conjecture''. Our main result confirms the expectations of Gao, Richter, Seymour and Melzer.

\begin{restatable}{theorem}{main}\label{th:main}
Any $k$-irreducible triangulation of an orientable surface $S$ of genus $g$ has at most $966 k^2 g = O(k^2 g)$ edges.
\end{restatable}

The bound $O(k^2 g)$ in Theorem~\ref{th:main} is tight up to the value of the constant which we have not tried to optimize: intuitively it can be achieved by taking $g$ connected sums of tori resembling $k$ by $k$ grids. We refer to Melzer~\cite[Section~6.2]{melzer2019k} for precise descriptions of families of $k$-irreducible triangulations of size $\Theta(k^2g)$. One interesting follow-up question would be to determine the correct constant in some asymptotic regimes in $g$ and $k$.

Note that an immediate application of the Euler formula yields similar bounds for the number of vertices and faces. 

\subparagraph*{Systems of curves.} Our proof of Theorem~\ref{th:main} is actually a consequence of a stronger result on system of curves that we introduce now. We investigate $k$-irreducible triangulations from the powerful perspective initiated by (de Graaf and) Schrijver in a series of papers (see e.g., ~\cite{schrijver1991decomposition,Sch92,de1995characterizing,de1997decomposition}), which consists in encoding the metric structure of an embedded graph using a system of curves. Note that this approach was already implicit in Steinitz's proof of his aforementioned theorem, as explained in Chang~\cite{chang2018tightening}.

For the case of $k$-irreducible triangulations, this is done as follows. One first observes that instead of considering shortest non-contractible closed walks, one could consider instead shortest non-contractible closed curves crossing the triangulation exactly at the vertices. Such curves are commonly called \define{nooses} in structural graph theory. If one defines the length of those curves as this number of intersections, one sees that any non-contractible closed walk in a triangulation can be deformed to a non-contractible noose of the same length and vice-versa, see Figure~\ref{fig:noose}, left. The length of a shortest non-contractible noose is called the \define{face-width} (also called \define{representativity} in, e.g.~\cite{gao1996irreducible}), and therefore we can equivalently define $k$-irreducibility as having face-width at least $k$ and being irreducible for this property.

The \define{medial graph} $M$ of a triangulation $T$ of a surface $S$ is the embedded graph in $S$ which has a vertex in the middle of each edge of $T$ and in each face $f$ of $T$ we add edges between vertices of $M$ that belong to consecutive edges of $f$. See Figure~\ref{fig:noose}, right for an illustration. Note that $M$ is a simple $4$-regular graph and that the length of a noose $c$ for $T$ is equal to half the number of transverse intersections it has with $M$. Because $M$ is $4$-regular, it can be considered as a \define{system of curves} $\mathcal{C}=(c_1,\ldots,c_p)$ whose formal definition is postponed to Section~\ref{SS:system-of-curves}. The \define{crossing number} of a system of curves $\mathcal{C}$, denoted by $\crossnum(\mathcal{C})$ is its number of intersections, and the \define{length} of a curve $c$ transverse to $\mathcal{C}$ is the number of intersections of $c$ and $\mathcal{C}$, denoted by $\crossnum(c,\mathcal{C})$.

\begin{figure}
    \centering
    \includegraphics[height=3cm]{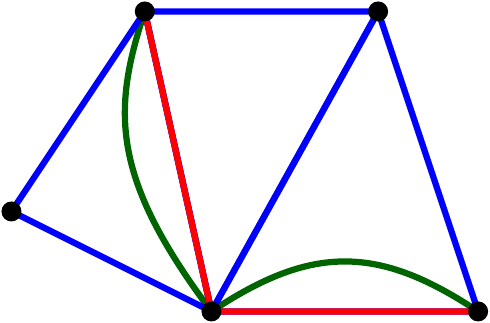}\hspace{0.5cm}\includegraphics[height=3cm]{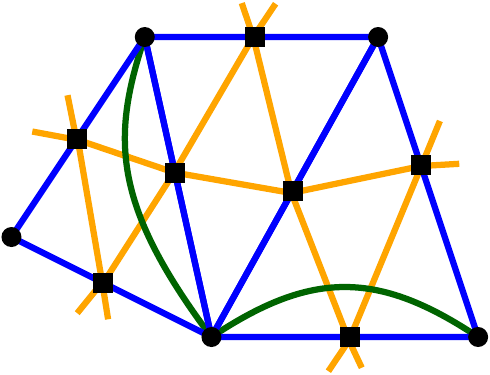}
    \caption{The triangulation in blue, a walk in red, its corresponding noose in green and the medial graph in orange.}
    \label{fig:noose}
\end{figure}

The strength of this perspective is that this allows to understand metric deformations of $T$ (e.g., edge contractions) under the lens of topological deformations (homotopies) of $\mathcal{C}$. For instance, it follows from the work of Schrijver~\cite{Sch92} that on an orientable surface, the system of curves corresponding to a minor-minimal graph of face-width $k$ is necessarily in \define{minimal position}: it contains no contractible curves and its crossing number $\crossnum(\mathcal{C})$ is minimized over all families $\mathcal{C}'$ homotopic to $\mathcal{C}$ (we refer to Section~\ref{S:prelim} for the background on homotopies and a precise statement of Theorem~\ref{th:schrijver2} by Schrijver).

A \define{smoothing} of a system of curves $\mathcal C$ is the system of curves obtained from $\mathcal C$ by removing a vertex $v$ and reconnecting the incident strands in one of two ways as in Figure~\ref{fig:Opening}. Note that a smoothing may change the number of components. 

\begin{figure}[ht!]
    \centering
	\includegraphics[height=2cm]{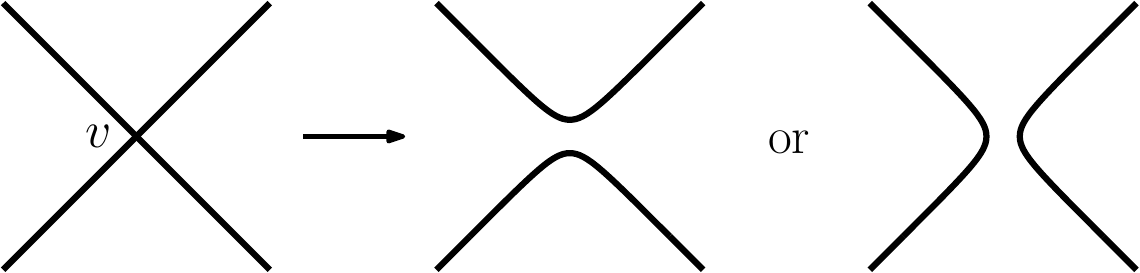}
	\caption{A smoothing at $v$}
	\label{fig:Opening}
\end{figure}

A \define{geodesic} is a closed curve transverse to $\mathcal{C}$ and of minimal length within its homotopy class. For a given system of curves $\mathcal{C}$, we say that a smoothing $\mathcal{C}'$ of $\mathcal{C}$ is \define{geodesically $k$-covered} if
\begin{itemize}
    \item there is a geodesic $c$ of length at most $k$ with respect to $\mathcal{C}$, and
    \item there is a curve $c'$ homotopic to $c$ and such that $\crossnum(c',\mathcal{C'})<\crossnum(c,\mathcal{C})$.
\end{itemize}
We say that a system of curves $\mathcal{C}$ on a surface $S$ is \define{geodesically $k$-covered} for some $k\ge 1$ if all its smoothings are geodesically $k$-covered and it does not have a connected component consisting of exactly one contractible curve. We now have all the ingredients to state our second theorem.

\begin{restatable}{theorem}{kirreduciblesystem}
\label{th:size-k-irreducible-system}
Let $\mathcal C$ be a geodesically $k$-covered system of curves on an orientable surface of genus $g$. The crossing number of $\mathcal C$ is at most $120 k^2 g = O(k^2g)$.
\end{restatable}

The proof technique we use allows us to treat the simplest case of \emph{non-orientable surfaces}.
\begin{restatable}{theorem}{nonorientable}
\label{t:projectiveplane}
Any $k$-irreducible triangulation of the projective plane has at most $4k^2$ edges.
\end{restatable}
However, there are some subtleties in the case of general non-orientable surfaces that prevent us from adapting our proof. See Section~\ref{S:nonorientable} for a discussion.

Finally, we remark that our tools could also be used to obtain slight strengthening of Theorem~\ref{th:main} and Theorem~\ref{t:projectiveplane}, yielding $O(k^2g)$ upper bounds under the weaker hypothesis that there exists a geodesic\footnote{Here, a geodesic is defined similarly to the geodesics for systems of curves: it is a closed walk in the triangulation having minimal length within its homotopy class.} of length at most $k$ going through every edge.

\subparagraph*{Connections to systolic geometry.} Theorems~\ref{th:main} and~\ref{th:size-k-irreducible-system} provide discrete counterparts to results in systolic geometry. This subfield of differential geometry investigates topological spaces endowed with continuous metrics (typically Riemannian) and aims at understanding the structure and properties of their shortest non-contractible curves, called \define{systoles}. We refer to the book of Katz~\cite{katz2007systolic} for an introduction. 

The famed systolic inequality of Gromov~\cite{gromov1983filling} states that for any surface of genus $g$ and area $A$, the systole has length $O(\sqrt{A/g} \log g)$. This is known to be tight: Brooks~\cite{brooks1988injectivity} and Buser and Sarnak~\cite{buser1994period} provided families of hyperbolic surfaces achieving this bound. There are discrete analogues of these results for embedded graphs~\cite{hutchinson1988short,colin2015discrete,kowalick2019filling}, showing that the edge-width of a triangulation with $n$ triangles is similarly bounded by $O(\sqrt{n/g} \log g)$ and that this is asymptotically tight. Equivalently, this shows that any triangulation of edge-width $k$ has $\Omega(k^2g/\log^2 g)$ triangles\footnote{Note that the growth is sensibly different if one considers the number of vertices instead of the number of triangles, see Melzer~\cite[Section~2.2]{melzer2019k}}. Bounding the size of $k$-irreducible triangulations then amounts to looking for a \emph{reverse systolic inequality}, bounding how far from this lower bound a locally extremal object for the systolic inequality can be. From this perspective, it is noteworthy that the logarithmic factor, which is known to be required for the systolic inequality, disappears in our bounds in Theorems~\ref{th:main} and~\ref{th:size-k-irreducible-system}, even though our bounds are also tight.

We are not aware of any works investigating continuous versions of such a reverse systolic inequality. One difficulty is that the space of Riemannian metrics on a surface is not compact, and extremal surfaces for the systolic inequality in negative Euler characteristic are expected to exhibit singularities, see for example Gromov~\cite[pp.62-65]{gromov1983filling}. This can be addressed by adding additional assumptions on the curvature, restricting to hyperbolic surfaces as in Fortier Bourque and Rafi~\cite{bourque2021local}, or nonpositively curved surfaces as in Katz and Sabourau~\cite{katz2021systolically}. Since the discrete setting of triangulations does not suffer from these analytic issues, it provides a convenient testbed for the study of locally extremal surfaces. The framework of geodesic currents of Bonahon~\cite{bonahon1988geometry} could provide an interesting middle-ground between these two worlds, keeping a continuous flavour but with strong compactness properties.

\subparagraph*{Related works.} Many articles in computational geometry in the last twenty years have been targeted at understanding topologically meaningful curves on surfaces, from the perspectives of both algorithms and combinatorics. We refer to the dedicated chapter by \'Eric Colin de Verdi\`ere in the handbook of computational geometry~\cite{c-ctgs-18} for a survey of this active field of research. Our proof techniques deal specifically with families of curves in minimal position on surfaces, a problem for which polynomial-time algorithms have been recently provided by Chang and de Mesmay~\cite{chang2022tightening}, and Dubois~\cite{dubois2024making}. The study of systolic properties of surfaces through the lens of such systems of curves has been instrumental in recent works of Cossarini~\cite{cossarini2020discrete} and Cossarini and Sabourau~\cite{cossarini2023minimal}. The medial graph construction connects homotopy of curves with electrical moves, as studied in Chang and Erickson~\cite{chang2017untangling}, Chang, Cossarini and Erickson~\cite{chang2019lower} and Aranguri, Chang and Fridman~\cite{aranguri2022untangling}. Very recently, Delecroix, Fontaine and Lazarus~\cite{delecroix2025computation} have made algorithmic Schrijver's concept of \emph{kernel}~\cite{Sch92}, which generalizes our $k$-irreducible triangulations to metric structures which are extremal with respect to the entire set of lengths of non-contractible curves (the \emph{length spectrum}).

\subparagraph*{Organization of the paper.} After some preliminaries in Section~\ref{S:prelim}, we first explain how Theorem~\ref{th:size-k-irreducible-system} implies Theorem~\ref{th:main} in Section~\ref{S:irred}. Then the main ingredient is Proposition~\ref{pr:filling-systoles} proven in Section~\ref{S:results}. Namely given a system of curves that is geodesically $k$-covered, it shows the existence of $O(g)$ curves of length at most $2k$ filling $S$. This immediately yields a first bound (Theorem~\ref{th:vertices-number-systolic}), albeit one not linear in $g$. Section~\ref{S:linear} introduces additional combinatorial tools to obtain the tight bound of Theorem~\ref{th:size-k-irreducible-system}, and Section~\ref{S:nonorientable} discusses the case of non-orientable surfaces. 

\section{Preliminaries}\label{S:prelim}
In this section we provide the main definitions used throughout the paper. The central objects of study are topological surfaces and curves on them whose definitions are recalled in Section~\ref{SS:surfaces-and-curves}. A (discrete) geometry on a surface allowing to measure lengths can then be defined in two related ways: either by the means of an embedded graph also called combinatorial surface (see Section~\ref{SS:combinatorial-surfaces}), or using a system of curves as in Section~\ref{SS:system-of-curves}.

\subsection{Surfaces and curves}\label{SS:surfaces-and-curves}
 In this paper, we use the term \define{surface} for a closed orientable topological surface, except in Section~\ref{S:nonorientable} where we deal with non-orientable surfaces. It is always denoted by $S$ and is sometimes implicit. We refer to Massey~\cite{Mas91} for topological background.

\subparagraph*{Curves.}  A \define{curve} on a surface $S$ is a continuous map $c:[0,1]\rightarrow S$ and a \define{closed curve} on $S$ is a continuous map $c: \mathbb{S}^1 \to S$, where $\mathbb{S}^1$ is the one-dimensional circle $\mathbb{R} / \mathbb{Z}$. Since all the curves in this article are closed, we will often just write curve instead of closed curve. Two closed curves $c$ and $c'$ on $S$ are \define{(freely) homotopic} if there is a continuous map $h:[0,1] \times \mathbb{S}^1 \rightarrow S$ such that $h(0,t)=c(t)$ and $h(1,t)=c'(t)$. Being freely homotopic is an equivalence relation on closed curves on $S$ and we refer to \define{(free) homotopy classes} for an equivalence class of closed curves with respect to free homotopy. We always denote curves with latin letters and homotopy classes with greek letters. A closed curve is \define{contractible} if it is homotopic to a constant curve. A curve is \define{simple} if it does not self-intersect. A closed curve is \define{primitive} if it is not homotopic to the $k$-fold concatenation of a curve $c$ for some $k\ge 2$. In particular, a primitive curve is not contractible.

\subsection{Combinatorial surfaces}\label{SS:combinatorial-surfaces}
\subparagraph*{Combinatorial surfaces.} A \define{combinatorial surface} is an embedding $\phi: G \to S$ of a graph $G=(V,E)$, possibly with multiple edges and loops, in a surface $S$, such that $S \setminus \phi(G)$ is a disjoint union of open disks. These disks are called the \define{faces} of $G$. The counterclockwise sequence of arcs (i.e., oriented edges) in the boundary of a face is called a \define{face boundary} or \define{facial walk}. The \define{degree of a face} is the length of its facial walk.

Two combinatorial surfaces $\phi: G \to S$ and $\phi': G' \to S'$ are \define{isomorphic} if there exists a graph isomorphism $\theta: G \to G'$ and a homeomorphism $\psi: S \to S'$ such that $\phi' \circ \theta = \psi \circ \phi$. The isomorphism class of a combinatorial surface can be conveniently encoded by a \define{rotation system} describing the circular orderings of the (oriented) edges around each vertex. See Mohar and Thomassen~\cite{mt-gs-01} for more details on this structure. Most of the time, the embedding $\phi$ is implicit and we identify $G$ and the corresponding combinatorial surface.

Let $\phi: G \to S$ be a combinatorial surface. The \define{length (with respect to $G$)} of a curve $c$ in $S$ that crosses only the vertices of $G$ and not its edges is the number of times this curve passes through a vertex of $G$. It is denoted $\crossnum(c, G)$. The \define{systole} of $G$ is the minimum length over all non-contractible closed curves $c$ on $S$ passing only through vertices of $G$ of $\crossnum(c, G)$. We will use the common abuse of language to also call systole such a non-contractible closed curve of minimal length when it leads to no confusion.

\subparagraph*{Triangulation.} A combinatorial surface $(S,G)$ is a \define{triangulation} if $G$ is simple (no loops nor multiple edges)  and every face has degree $3$.
A triangulation $T$ on a surface $S$ is \define{$k$-irreducible} if
\begin{enumerate}
    \item the systole of $T$ is equal to $k$

    \item for every edge $e=(uv)$ of $T$, there is a non-contractible curve $c$ of length $k$ on $T$ crossing consecutively $u$ and $v$.
\end{enumerate}

\begin{remark}
In line with the introduction, we define our triangulations to be simplicial, but note that this is implied by $k$-irreducibility for $k \ge 3$. Indeed, if there is a loop or a bigon formed of two parallel edges, it has to be contractible and thus bounds a disk, but then the edges inside that disk cannot belong to a systole as there would be a shortcut, violating $k$-irreducibility.
\end{remark}

\subsection{Systems of curves}\label{SS:system-of-curves}

\subparagraph*{Systems of curves.}

A \define{system of curves} on a surface $S$ is a family $\mathcal C=(c_1,\ldots, c_p)$ of non-contractible closed curves on $S$ in \emph{general position}, i.e., they only cross transversely and at most two curves cross at a point.

An \define{intersection} of $\mathcal C$ is a pair $(i,s) \not= (j,t)$ such that $c_i(s) = c_j(t)$. Because a system of curves is in general position, the set of intersections is an isolated subset of points in $S$. The \define{crossing number of $\mathcal C$} is its number of intersections. Note that $\mathcal C$ defines a $4$-regular embedded graph on $S$ whose vertex set is the intersections of $\mathcal C$ and its edges are segments of the curves $c_i$. Conversely, a $4$-regular embedded graph $\phi: G \to S$ determines a system of curves by `going straight' at each crossing. We will freely use this identification and use the vocabulary of embedded graphs, i.e., vertices, edges and faces for a system of curves $\mathcal C$. Note however that in this context the embedded graph might not be a combinatorial surface because its complement could have topology. A system of curves is \define{connected} if the corresponding graph is connected.

Two systems of curves $\mathcal C=\{c_1,\ldots, c_p\}$ and $\mathcal C'=\{c'_1, \ldots, c'_p\}$ are \define{homotopic} if for all $i$ in $[1,p]$, $c_i$ is freely homotopic to $c'_i$. We say that a system of curves $\mathcal C$ is in \define{minimal position} if it contains no contractible curves and has the minimal crossing number over all systems of curves homotopic to $\mathcal C$. We say that two curves \define{intersect essentially} if they intersect at least once in minimal position.

\subparagraph*{Length of curves.} Let $\mathcal C$ be a system of curves on $S$ and $c$ be a curve transverse to $\mathcal C$. The \define{length} of $c$ with respect to $\mathcal C$ is the number of crossings of $c$ with $\mathcal C$. We recall that this quantity is denoted $\crossnum(c, \mathcal C)$. The \define{length} of a free homotopy class $\gamma$ is the minimal number of crossings of a representative $c$ of $\gamma$ with $\mathcal C$. We use the notation $\crossnum(\gamma, C)$ as a shortcut for $\min_{c \in \gamma} \crossnum(c, C)$. If the length of $c$ equals the length of its free homotopy class, we say that $c$ is a \define{geodesic} (equivalently $c$ is in minimal position with $\mathcal{C}$). For $\mathcal{C}$ in minimal position, we define the \define{systole} of $\mathcal C$ to be $\min_{\gamma} \crossnum(\gamma, \mathcal C)$ where $\gamma$ runs over all non-trivial free homotopy classes of closed curves in $S$. Here again, we sometimes also call a closed curve achieving this minimum length a systole. 
We say that a system of curves $\mathcal C$ in minimal position on a surface $S$ is \define{filling} if for any non-trivial homotopy class $\gamma$ we have $\crossnum(\gamma, \mathcal C) > 0$. Equivalently, a system of curves $\mathcal C$ in minimal position on a surface $S$ is filling if and only if it corresponds to a combinatorial surface (i.e. the complement of the curves is a disjoint union of topological disks).

\subparagraph*{Geodesically $k$-covered and $k$-irreducible systems of curves.} Recall from the introduction that a system of curves $\mathcal{C}$ on a surface $S$ is \define{geodesically $k$-covered} if all its smoothings are geodesically $k$-covered and it does not have a connected component consisting of exactly one contractible curve. We say that a system of curves is \define{$k$-irreducible} for some $k\ge 1$ if it is geodesically $k$-covered and every non-contractible closed curve $c$ on $S$ has length at least $k$. This mirrors the definition for triangulations.

A system of curves $\mathcal{C}$ is \define{tight} if it does not contain a contractible curve disjoint from the other curves and if for any smoothing $\mathcal C'$ of $\mathcal C$ there is a curve in minimal position with respect to $\mathcal C$ and $\mathcal C'$ which is shorter with respect to $\mathcal C'$ than with respect to $\mathcal C$.  The following is an (equivalent reformulation of) a notion introduced by Schrijver, who proved the following theorem:

\begin{theorem}[{\cite[Theorem~5]{schrijver1991decomposition}}]\label{th:schrijver2}
A system of curves on an orientable surface is tight if and only if it is made of primitive curves in minimal position.
\end{theorem}

It is immediate that a geodesically $k$-covered system of curves is tight, and therefore it is in minimal position. Conversely, a tight system of curves is geodesically $k$-covered for some $k$. Note also that a $k$-irreducible system of curves is necessarily filling. While a geodesically $k$-covered system of curves may not be filling, by cutting the surface along simple non-contractible closed curves that it misses, we can obtain a new surface of smaller genus in which it is filling, and the smaller genus will only improve our bounds. So without loss of generality we assume henceforth that geodesically $k$-covered systems of curves are filling.

\begin{remark}
In this article, we often alternate between different metric models. In order to minimize confusion, we use the following guidelines:
\begin{itemize}
    \item in a triangulation $T$, we consider curves which are walks on the edges of the triangulation, and their lengths are the number of edges that they use.
    \item in a combinatorial surface $G$, we consider curves which cross $G$ only at vertices, and their lengths are the number of vertices that they cross.
    \item in a system of curves $\mathcal{C}$, we consider curves which are in general position with respect to $\mathcal{C}$, and their lengths are their number of intersections with $\mathcal{C}$.
\end{itemize}
\end{remark} 

\section{From triangulations to systems of curves}\label{S:irred}
The goal of this section is to prove how Theorem~\ref{th:main} on $k$-irreducible triangulations follows from Theorem~\ref{th:size-k-irreducible-system} on geodesically $k$-covered systems of curves.

Let $T$ be a $k$-irreducible triangulation on $S$ with vertex set $V$. We first explain how to associate a system of curves to $T$, and we then study its properties. Let $G$ be a subgraph of $T$ with the same vertex set $V$ and minimal (with respect to edge inclusion) for the property that every non-contractible curve $c$ intersecting $G$ only at $V$ satisfies $\crossnum(c,G) \ge k$. Let $M$ be the medial graph of $G$: recall that it is the graph with a vertex at the midpoint of every edge of $G$ and two midpoints are connected by an edge of $M$ whenever their supporting edges in $G$ are consecutive in a face of $G$, see Figure~\ref{Medial}. The graph $M$ is $4$-regular and is thus naturally associated to a system of curves $\mathcal{C}$ on $S$. We call $\mathcal C$ a \define{system of curves associated to $T$} (note that there can be many of them, depending on the choice of $G$). As proven in~\cite{Sch92}, for any homotopy class $\gamma$ of curves in $S$ we have $\crossnum(\gamma, \mathcal C) = 2 \crossnum(\gamma, G)$.

\begin{figure}
    \centering
    \includegraphics[height=3.5cm]{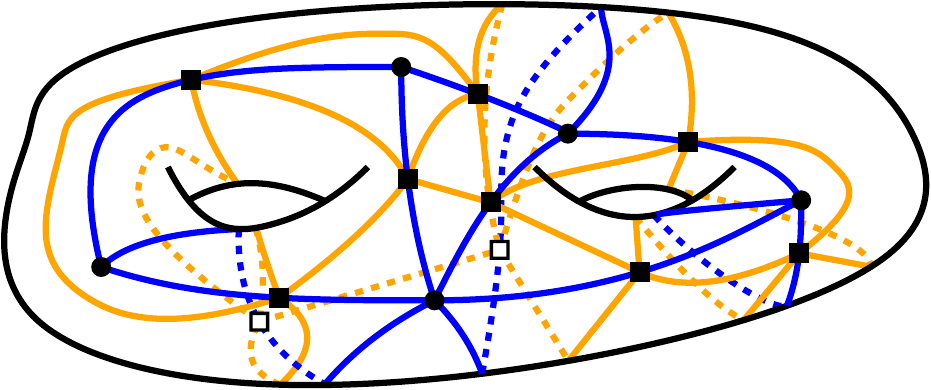}
    \caption{A graph $G$ in blue on a genus $2$ surface and its medial graph in orange}
    \label{Medial}
\end{figure}

\begin{proposition}
\label{pr:irreducible-triangulation-system}
Let $T$ be a $k$-irreducible triangulation on $S$. Then any system of curves associated to $T$ is $2k$-irreducible and in minimal position.
\end{proposition}

\begin{proof}
Let us denote as before by $G$ an (edge inclusion-wise) minimal subgraph of $T$ such that the length of a shortest non-contractible noose on $G$ is still $k$, and by $\mathcal C$ the system of curves associated to $G$. As $\crossnum(c,G)\ge k$ for every non-contractible curve $c$, $\crossnum(\gamma, \mathcal{C}) = 2 \crossnum(\gamma,G) \ge 2k$ for every non-trivial homotopy class $\gamma$.

Let $\mathcal{C}'$ be a smoothing of $\mathcal{C}$. As shown in~\cite{Sch92}, a smoothing of $\mathcal C$ either corresponds to an edge deletion or an edge contraction in $G$. We work separately on these two cases.

\begin{enumerate}
    \item Assume that $\mathcal{C}'$ is obtained from $\mathcal{C}$ by a smoothing corresponding to deleting the edge $e$. Let $G'=G\setminus \{e\}$. Then by minimality of $G$ in its definition, there is a non-contractible curve $c$ such that $\crossnum(c, G')<k$. Thus $\crossnum(\gamma, \mathcal{C}')\le 2\cdot \crossnum(c, G')<2k$ where $\gamma$ is the homotopy class of $c$.
    
    \item Assume that $\mathcal{C}'$ is obtained from $\mathcal{C}$ by a smoothing corresponding to contracting the edge $e=(uv)$. As $T$ is a $k$-irreducible triangulation, there is a systole $c$ of $T$ going through $u$ and $v$ consecutively. Thus this systole is also a systole of $G$ with the same property and is also a systole of $\mathcal{C}$ crossing two consecutive edges of $\mathcal{C}$ around the vertex $v(e)$ (see Figure \ref{fig:edge-contraction}). Thus, the smoothing reduces the length of $c$ and $\crossnum(\gamma, \mathcal{C}')<2k$ where $\gamma$ denotes the homotopy class of $c$.
\end{enumerate}

Since $G$ is connected, so is $\mathcal{C}$, and thus it does not have a connected component consisting of exactly one contractible curve. Therefore the system of curves $\mathcal{C}$ is $2k$-irreducible. It follows from Theorem~\ref{th:schrijver2} that it is in minimal position.
\end{proof}

\begin{figure}
    \centering
    \includegraphics{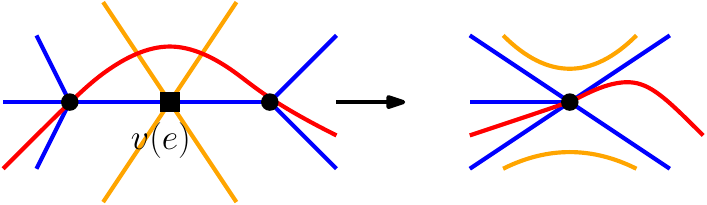}
    \caption{The contraction of an edge in $G$. The blue graph is $G$, the orange one is $\mathcal C$ and the red curve is $c$.}
    \label{fig:edge-contraction}
\end{figure}

We are now ready to prove Theorem~\ref{th:main} assuming Theorem~\ref{th:size-k-irreducible-system}. We restate it for convenience.

\main*

\begin{proof}
As before, let $G$ be a minimal subgraph of $T$ such that the length of a shortest non-contractible noose on $G$ is $k$, and let $\mathcal C$ be its associated system of curves. By Proposition~\ref{pr:irreducible-triangulation-system}, any system of curves associated to $T$ is a $2k$-irreducible system of curves, and thus it is geodesically $2k$-covered. Thus by Theorem~\ref{th:size-k-irreducible-system}, the crossing number $\crossnum(\mathcal{C})$ of $\mathcal{C}$ is at most $120(2k)^2g=480k^2g$.

Now let $M$ be the medial graph of $G$ from which the system of curves $\mathcal{C}$ is built. It is such that $|V(M)| = \crossnum(\mathcal{C})$ and because it is quartic $|E(M)| = 2 \crossnum(\mathcal{C})$. The graph $M$ is face bipartite and the classes of the bipartition are in bijection with the vertices $V(G)$ and the faces $F(G)$ of $G$. These bijections preserve the degree and hence
$\sum_{v \in V(G)} \deg_G(v) = |E(M)| = 2 \crossnum(\mathcal{C})$. Now, by Proposition~\ref{pr:irreducible-triangulation-system}, $\mathcal{C}$ is in minimal position. In particular, $\mathcal{C}$ has neither monogon nor bigon. In other words, each face of $M$ has degree at least $3$. Thus $\sum_{v \in V(G)} \deg_G(v) \ge 3 |V(G)|$. We conclude that $|V(T)| = |V(G)| \le \frac{2}{3} \crossnum(\mathcal{C}) \le \frac{2}{3} 480 k^2 g = 320 k^2 g$,  and $|E(T)| \leq 3|V(T)|+6g-6 \leq 966 k^2 g$.
\end{proof}

\section{Short filling systems for geodesically \boldmath\texorpdfstring{$k$}{k}-covered systems of curves}\label{S:results}

The main result of this section is the following proposition, which extracts from the property of being geodesically $k$-covered a more topological property: the existence of a filling family of short curves.

\begin{proposition}
\label{pr:filling-systoles}
Let $\mathcal C$ be a geodesically $k$-covered system of curves. Then there is a set of curves of lengths at most $2k$ with respect to $\mathcal C$ which is in minimal position and filling.
\end{proposition}

The proof of this proposition relies partially on the notion of universal cover. We refer to Massey~\cite[Chapter~5]{Mas91} for an introduction to this classical object in algebraic topology and to Farb and Margalit~\cite[Section~1.2]{farb2011primer} for the specifics in the case of curves on surfaces. Given a surface $S$ different from the sphere, its \emph{universal cover} $\tilde{S}$ is homeomorphic to the plane $\mathbb{R}^2$. Given a curve $c: \mathbb{S}^1 \to S$ in a surface $S$, a \emph{lift} $\tilde{c}: \mathbb{R} \to \tilde{S}$ of $c$ is a map to the universal cover $\tilde{S}$ that commutes with the projections $\mathbb{R} \to \mathbb{S}^1$ and $\tilde{S} \to S$. When the lift $\tilde{c}$ is simple we call it a \emph{line}.

\begin{proof}
We will build a set of curves of lengths at most $2k$ made of geodesics of length at most $k$ and some curves that are obtained by concatenation of two of these, and then show that these curves are filling.

The following claim is standard in geometry of curves on surfaces. As it is used many times throughout our proof, we provide a standalone proof.

\begin{claim}\label{cl:intersections-cover}
Let $c$ and $c'$ be non-contractible curves on $S$ in minimal position. Then any two lifts $\tilde{c}$ and $\tilde{c'}$ in the universal cover $\tilde{S}$ are lines that intersect at most once.
\end{claim}

\begin{proof}
\renewcommand\qedsymbol{$\lrcorner$}
Assume that the lift $\tilde{c}$ has a self-intersection. Then, by the Jordan curve theorem in $\tilde{S}$, it forms a monogon: there is a disk in $\tilde{S}$ bounded by a subline of $\tilde{c}$. By homotopy, one can remove this monogon contradicting the fact that $c$ was in minimal position on $S$. Therefore, $\tilde{c}$ and $\tilde{c'}$ are lines in $\tilde{S}$.

Let us assume for contradiction that two lifts $\tilde{c}$ and $\tilde{c'}$ intersect twice. Then, again by the Jordan curve theorem, they form a bigon: there is a disk in $\tilde{S}$ bounded by one subline $\alpha$ of $\tilde{c}$ and one subline $\beta$ of $\tilde{c'}$. Now, the homotopy moving $\alpha$ to $\beta$ projects to a homotopy on the surface, and leads to a new pair of curves intersecting less. This contradicts the hypothesis that $c$ and $c'$ were in minimal position.
\end{proof}

The next claim is the first step in defining the curves we will work with.

\begin{claim}
\label{cl:filling-vertex-property}
Let $v$ be a crossing of $\mathcal C$ and $e$, $e'$ two edges adjacent to $v$ that are consecutive for the counter-clockwise ordering of edges around $v$. Then there is a geodesic of length at most $k$ that passes consecutively through $e$ and $e'$.
\end{claim}

\begin{proof}
\renewcommand\qedsymbol{$\lrcorner$}
Let $\mathcal C'$ be the system of curves obtained by smoothing $v$ and connecting $e$ with $e'$. As $\mathcal C$ is geodesically $k$-covered, there exist homotopic geodesics $c$ and $c'$ respectively for  $\mathcal C$ and $\mathcal C'$ and such that $\crossnum(c,\mathcal C') < \crossnum(c,\mathcal C)\leq k$. Since $\mathcal C$ and $\mathcal C'$ differ from each other only in the neighborhood of $v$ in $S$ where the smoothing has been performed, we may assume that $c=c'$ and the curve $c$ must use at least once the gap between $e$ and $e'$ formed by the smoothing. Therefore, $c$ crosses consecutively $e$ and $e'$ in a neighborhood of $v$. 
\end{proof}

Let $e_1,e_2,e_3$ and $e_4$ be the four consecutive edges around a vertex $v$ of $\mathcal{C}$. By Claim~\ref{cl:filling-vertex-property} there is a non-contractible closed curve $s_1$ in minimal position with $\mathcal C$ crossing consecutively $e_1$ and $e_2$. Similarly, there is a non-contractible closed curve $s_2$ in minimal position with $\mathcal C$ crossing consecutively $e_2$ and $e_3$. We denote by $\mathcal{S}$ a set of curves containing a pair of such curves $s_1$ and $s_2$ for each vertex of $\mathcal{C}$. 

We first handle the case where the surface is a torus. 

\begin{claim} \label{cl:filling-vertex-property2}
If the surface $S$ is a torus and $v$ is a vertex of $\mathcal{C}$ and $s_1$ and $s_2$ are the two corresponding closed curves of $\mathcal{S}$, then $s_1$ and $s_2$ intersect essentially.
\end{claim}

\begin{proof}[Proof of Claim~\ref{cl:filling-vertex-property2}]
Because $\mathcal C$ is in particular tight, all the curves it contains are primitive by Theorem~\ref{th:schrijver2}. On the torus, this is equivalent to saying that they are simple curves. We denote by $c_1$ and $c_2$ the two curves of $\mathcal C$ containing respectively $e_1$ and $e_2$. The universal cover of the torus is $\mathbb{R}^2$, any pair of lifts $\tilde{c_1}$ and $\tilde{c_2}$ of respectively $c_1$ and $c_2$ are lines that cross exactly once (because $\mathcal{C}$ is in minimal position). In particular, $c_1$ and $c_2$ are not homotopic to each other. Likewise, all other curves of $\mathcal{C}$ lift to lines crossing $c_1$ and $c_2$ at most once. We lift $s_1$ and $s_2$ into $\tilde{s_1}$ and $\tilde{s_2}$, choosing a basepoint for the lifting in a face adjacent to this crossing point $x$. Since $s_1$ and $s_2$ are in minimal position with $\mathcal{C}$, each lift $\tilde{s_i}$ crosses each lift $\tilde{c_j}$ at most once by Claim~\ref{cl:intersections-cover}.

We fix an orientation of the torus and of $c_1$ and $c_2$ and orient $s_1$ and $s_2$ so that $\tilde{s_1}$ and $\tilde{s_2}$ both intersect $\tilde{c_1}$ positively: this means that the local orientation of $(x,s_i,c_1)$ matches the chosen orientation of the torus. Then, by construction $\tilde{s_1}$ and $\tilde{s_2}$ intersect $\tilde{c_2}$ with different signs. This implies that $s_1$ and $s_2$ are not homotopic as this sign is preserved under homotopies, and reversing the orientation of one curve to fix this would lead to opposite signs with $\tilde{c_1}$. 

On a torus, the only possible way for two curves to not intersect essentially is if they are both homotopic to a power of the same primitive curve. This would be incompatible with the signs of the intersections with the curves $c_i$, concluding the proof of the claim.
\end{proof}

Since on a torus, any pair of curves intersecting essentially is filling, Claim~\ref{cl:filling-vertex-property2} concludes the proof for the torus with the family $\mathcal{S}$. In the rest of the proof we assume that the surface has negative Euler characteristic. Therefore its universal cover is homeomorphic to an open disk and can be endowed with a boundary at infinity $\partial \tilde{S}$ which is a circle. The lifts of curves in $\tilde{S}$ are lines with endpoints on this boundary, and two curves are homotopic if and only if there exist lifts of the two curves having the same endpoints. Furthermore, no two lifts have exactly one endpoint in common. We refer to~\cite[Section~1]{farb2011primer}.

\begin{figure}
    \centering
    \includegraphics[height=4cm]{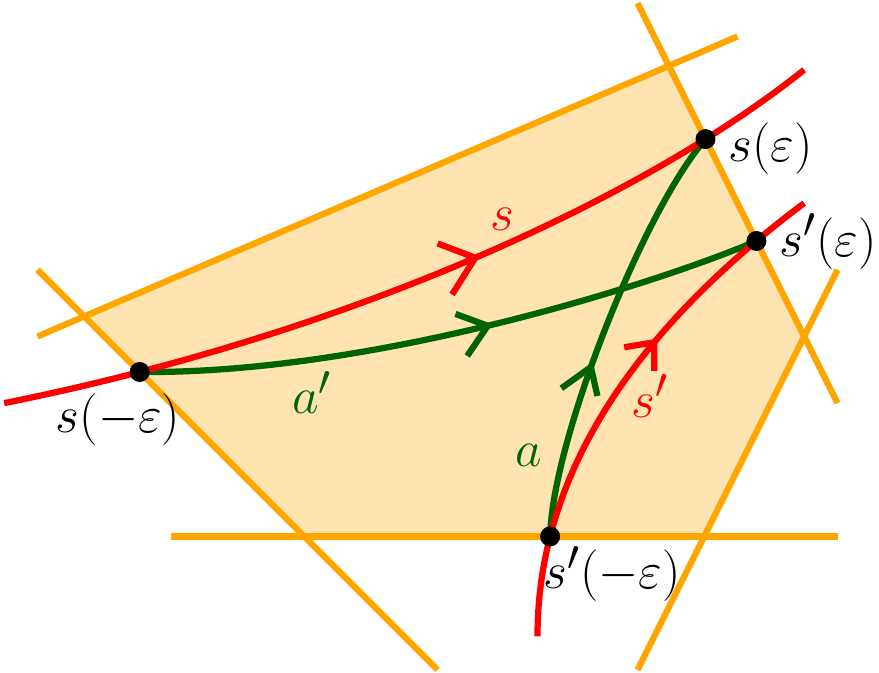}
    \caption{Wedding of two curves $s$ and $s'$.}
    \label{fig:definition_t_uv}
\end{figure}

We now define curves obtained by ``gluing'' or ``unsmoothing'' two curves in $\mathcal{S}$. They will be used to fill the surface $S$. Let $s$ and $s'$ be two curves in $\mathcal S$ that start in the same face $f$ of $\mathcal C$ and such that $\mathcal C \cup \{s, s'\}$ is in minimal position. Up to homotopy we can further assume that the first crossings with the boundary of $f$ occur for both $s$ and $s'$ at time $\pm\epsilon$ and that $s$ and $s'$ do not intersect in this initial face. Namely $s([-\epsilon,\epsilon]) \cap s'([-\epsilon,\epsilon]) = \emptyset$ and that the points $s(\epsilon)$, $s(-\epsilon)$, $s'(-\epsilon)$ and $s'(\epsilon)$ appear on the boundary of $f$ in that counterclockwise order. Let $a$ and $a'$ be arcs in the face $f$ that are in minimal position and go respectively from $s'(-\epsilon)$ to $s(\epsilon)$ and $s(-\epsilon)$ to $s'(\epsilon)$. These two arcs intersect exactly once. We define a new curve $t$ by the concatenation $t := a \cdot s|_{[\epsilon,-\epsilon]} \cdot a' \cdot s'|_{[\epsilon,-\epsilon]}$, see Figure~\ref{fig:definition_t_uv}.
We say that $t$ is obtained by \define{wedding} $s$ and $s'$ if $t$ is in minimal position with respect to itself. Equivalently, we can make a wedding if in the universal cover $\widetilde S$ of $S$ the two lifts of $s$ and $s'$ starting from the same lift $\widetilde f$ of $f$ do not intersect. By definition, a wedding of two curves in $\mathcal{S}$ has length at most $2k$ with respect to $\mathcal C$. If such a wedding is put in minimal position with respect to $\mathcal C$ by a homotopy, it might even get shorter.

We define the set $\mathcal S_2$ to be a set of representatives of the homotopy classes of all possible weddings of curves in $\mathcal S$ such that $\mathcal C \cup \mathcal S \cup \mathcal S_2$ is in minimal position. The following claim proves that the curves in $\mathcal S \cup \mathcal S_2$ intersect every possible simple closed curve on the surface.

\begin{claim}
\label{cl:filling-curve-property}
Let $c$ be a non-contractible simple closed curve. Then there exists $s \in (\mathcal S \cup \mathcal S_2)$ such that $c$ and $s$ intersect essentially.
\end{claim}

\begin{proof}
\renewcommand\qedsymbol{$\lrcorner$}
We first put $c$ in minimal position with respect to $\mathcal{C} \cup \mathcal S \cup \mathcal S_2$. Since $\mathcal{C}$ is filling, there is an edge $e=(uv)$ of $\mathcal{C}$ that intersects $c$. Let $c_1$ be the curve of $\mathcal C$ containing the edge $e$ and $c_2$ and $c_3$ be the curve crossing $c_1$ at $u$ and $v$ respectively. We denote by $s_1$ and $s_2$ the two curves in $\mathcal S$ corresponding to $u$ and by $s_3$ and $s_4$ the ones corresponding to $v$. We will prove that $c$ intersects either $s_1$, $s_2$, $s_3$, $s_4$ or one of their pairwise weddings. 

\begin{figure}
    \centering
    \includegraphics[height=5cm]{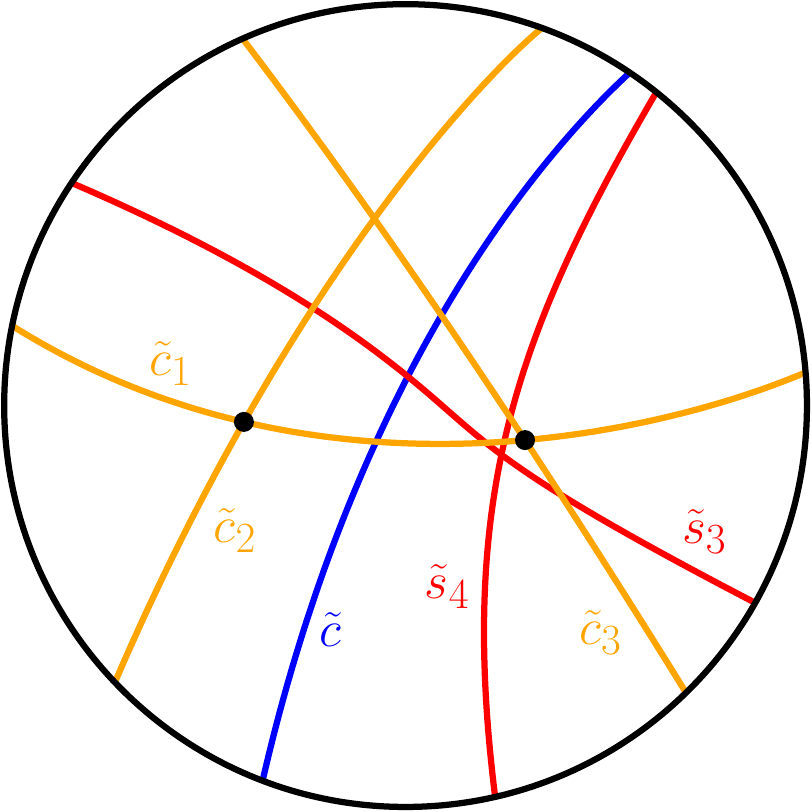}\hspace{1cm}\includegraphics[height=5cm]{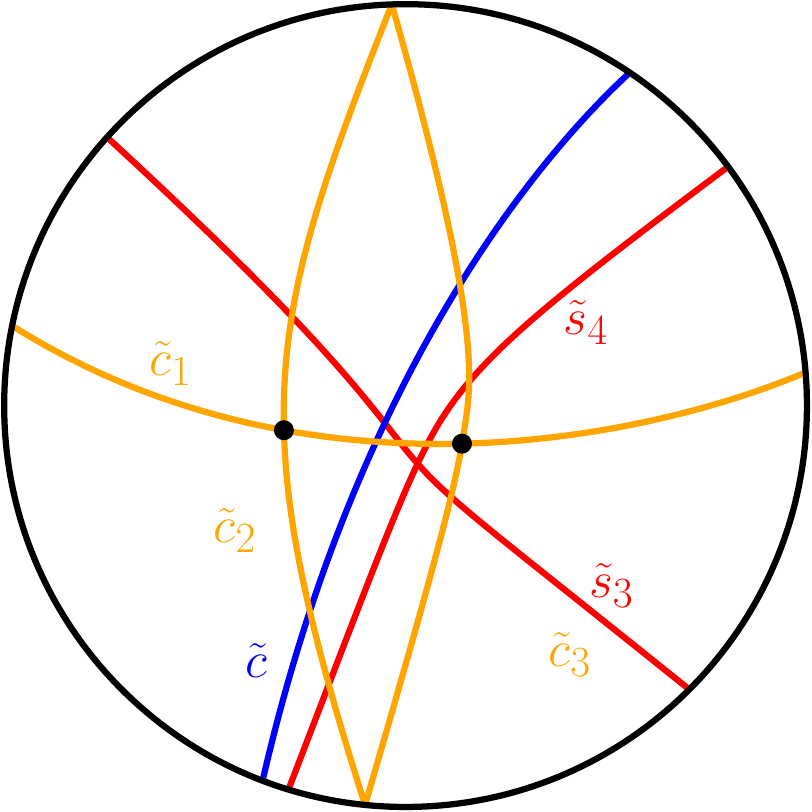}\hspace{1cm}\includegraphics[height=5cm]{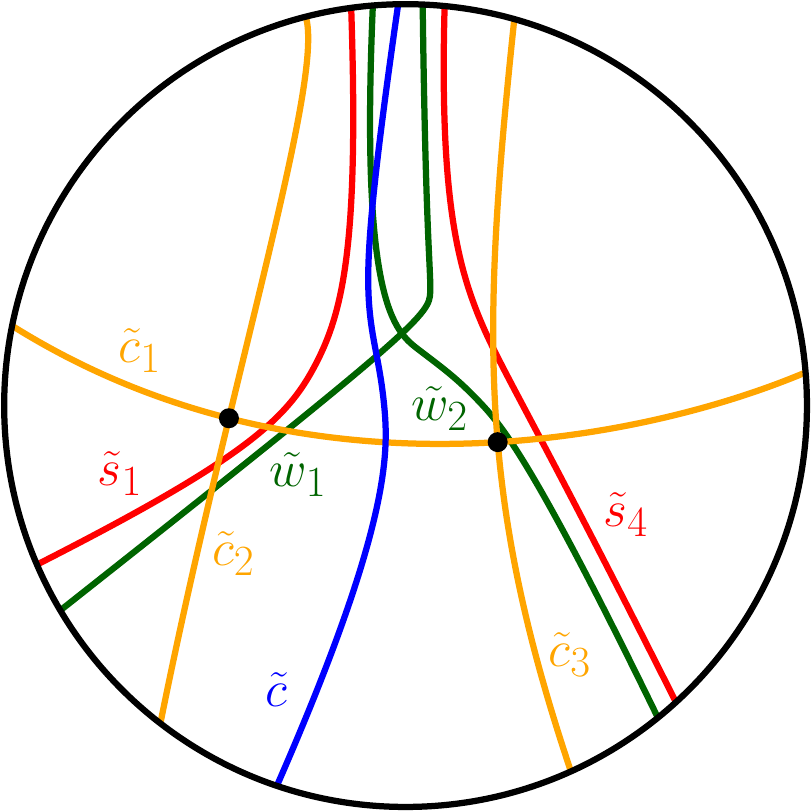}
    \caption{The three different cases leading to the proof of claim \ref{cl:filling-curve-property}.}
    \label{fig:intersect_cases}
\end{figure}

In the universal cover, our local picture lifts as in Figure~\ref{fig:intersect_cases}, yielding lines $\tilde{c_1}$, $\tilde{c_2}$, $\tilde{c_3}$, $\tilde{s_1}$, $\tilde{s_2}$, $\tilde{s_3}$, $\tilde{s_4}$ and $\tilde{c}$. The lines $\tilde{c_1}$,$\tilde{c_2}$, $\tilde{c_3}$ intersect pairwise at most once and are simple because $\mathcal{C}$ is in minimal position. We do a case by case analysis depending on the topology of $\tilde{c_2}$ and $\tilde{c_3}$. The only ingredient we use is the fact that each pair of these curves intersect at most once.

If $\tilde{c_2}$ and $\tilde{c_3}$ intersect, then it is immediate that $\tilde{c}$ intersects at least one of the curves $\tilde{s_1}$, $\tilde{s_2}$, $\tilde{s_3}$ or $\tilde{s_4}$, see top left Figure~\ref{fig:intersect_cases}. The same happens if $c_2$ and $c_3$ are homotopic since in that case $\tilde{c_2}$ and $\tilde{c_3}$ have the same endpoints in $\partial \tilde{S}$. Since these intersections are forced by the endpoints on $\partial \tilde{S}$ and homotopies do not move these endpoints, these intersections on $\tilde{S}$ project to essential intersections on $S$. See top right Figure~\ref{fig:intersect_cases}.

The last case is if $\tilde{c_2}$ and $\tilde{c_3}$ are disjoint (including at their endpoints) as in bottom Figure~\ref{fig:intersect_cases}. Then the wedding $w$ of $s_1$ and $s_4$ based at the added crossing lifts to two lines $\tilde{w_1}$ and $\tilde{w_2}$ which cross the lines $\tilde{c_1}$ and $\tilde{c_2}$, respectively $\tilde{c_2}$ and $\tilde{c_3}$ the same way that $\tilde{s_1}$ and $\tilde{s_4}$ do. In particular, this wedding exists since $\tilde{w_1}$ and $\tilde{w_2}$ cannot cross again due to $\tilde{c_1}, \tilde{c_2}$ and $\tilde{c_3}$ acting as barriers for $\tilde{s_1}$ and $\tilde{s_4}$. Since $\tilde{c}$ intersects $\tilde{e}$, it must intersect either $\tilde{w_1}$ or $\tilde{w_2}$ and thus $c$ intersects $w$ essentially. This concludes the proof.
\end{proof}

Therefore, the family $\mathcal{S} \cup \mathcal{S}_2$ intersects every possible simple non-contractible curve in the surface: it is filling. This claim concludes the proof.
\end{proof}

We encapsulate this last result in the following definition: we say that a system of curves $\mathcal C$ is \define{$k$-saturating} if there exists another system of curves $\mathcal S$ such that:

\begin{itemize}
    \item $\mathcal{C} \cup \mathcal{S}$ is in minimal position,
    \item Each curve in $\mathcal{S}$ has length at most $k$, and
    \item The system of curves $\mathcal{S}$ is filling.
\end{itemize}

Therefore, Proposition~\ref{pr:filling-systoles} shows that a geodesically $k$-covered system of curves is $2k$-saturating. The following lemma shows that one can extract a filling family of size $O(g)$ from any filling family. This is certainly folklore, see for example Chen~\cite[Theorem~3.15]{chen2025indexgapsystolefunction}. We include a proof for completeness.

\begin{lemma}
\label{le:genus-size-filling-family}
Let $\mathcal C$ be a filling system of curves in minimal position on a closed orientable surface $S$ of genus $g \ge 1$. Then there is a sub-system $\mathcal{C}' \subset \mathcal C$ of at most $3g-1$ curves which is also filling.
\end{lemma}

\begin{proof}[Proof of Lemma~\ref{le:genus-size-filling-family}]
For $\mathcal{C}$ a system of curves in minimal position, the \define{subsurface spanned by $\mathcal C$} is a closed subsurface $S(\mathcal{C}) \subseteq S$ defined by first taking a regular neighborhood $N(\mathcal C)$ of the curves in $\mathcal{C}$, and then filling the components of $S \setminus N(\mathcal C)$ which are topological disks.

We proceed greedily. We pick a curve $c_1 \in \mathcal{C}$ and let $\mathcal{C}_1 := \{c_1\}$. Then we define a finite sequence of curves and curve systems $\mathcal{C}_i$ for $i \ge 1$ as follows
\begin{itemize}
    \item if $S(\mathcal{C}_i)=S$, we stop
    \item if not, we pick $c_{i+1} \in \mathcal{C}$ such that $\crossnum(\mathcal C_i, c_{i+1}) > 0$ and there exists a closed curve $c$ in minimal position with respect to $\mathcal C$ such that $\crossnum(\mathcal{C}_i, c) = 0$ but $\crossnum(c_{i+1},c)>0$, and we set $\mathcal{C}_{i+1}=\mathcal{C}_i\cup\{c_{i+1}\}$.
\end{itemize}
Note that the sequence is well defined because $\mathcal{C}$ is filling. For each $i$, let $S_i$ be the subsurface spanned by $\mathcal{C}_i$. By definition $S_{i-1}\subset S_i$. By the first property in the construction of $\mathcal C_{i+1}$ from $\mathcal C_i$ it follows by induction that each $S_i$ is connected. By the second property, each inclusion $S_{i-1} \subset S_i$ is topologically strict in the sense that each $S_i$ contains at least one essential (non-contractible and not homotopic to a boundary) simple closed curve disjoint from $S_{i-1}$. Since the surface $S_i \setminus S_{i-1}$ has at least one boundary and is connected and orientable, this implies that it either has negative Euler characteristic, or is an annulus whose boundaries are not homotopic to a boundary of $S_i$. In the second case, by connectedness the genus of $S_{i-1}$ is smaller than the one of $S_i$. As the genus of $S_i$ is smaller than $g$, this case happens at most $g$ times. Finally, we recall that the Euler characteristic is additive: $\chi(S_{i-1})+\chi(S_i \setminus S_{i-1})=\chi(S_i)$. So since $\chi(S_1)\le 0$, the first case happens at most $2g-2$ times.

In total, there are at most $g+2g-2$ inclusions of subsurfaces, and thus at most $3g-1$ curves in the constructed system of curves when the algorithm stops.
\end{proof}

We can now combine our tools to obtain a first polynomial upper bound.

\begin{theorem}
\label{th:vertices-number-systolic}
Let $\mathcal C$ be a $k$-saturating system of curves. Then the crossing number of $\mathcal C$ is at most $2((3g-1) \cdot k)^2$.
\end{theorem}

\begin{proof}
By the definition of a $k$-saturating system of curves and thanks to Lemma~\ref{le:genus-size-filling-family}, there exists a filling family $\mathcal S = (s_i)_{i\in I}$ of closed curves of length at most $k$ with respect to $\mathcal{C}$ and of size at most $3g-1$. In particular, for every $i\in I$, $\crossnum(s_i, \mathcal C)=k$.

Recall that by the definition of $k$-saturating, the system $\mathcal{C} \cup \mathcal{S}$ is in minimal position. We think of $\mathcal{S}$ as a $4$-valent graph that is a combinatorial surface. In other words:
\begin{itemize}
    \item Every curve of $\mathcal C$ intersects transversely the curves $s_i$, and
    
    \item Every intersection vertex of $\mathcal C$ is inside a face of $\mathcal S$.
\end{itemize}
We denote by $F$ the faces of $\mathcal S$. For such a face $f$, let $p(f)$ denote the perimeter of $f$ i.e. the number of crossings of $\mathcal C$ with the boundary of $f$. Then, by double counting, $\sum_{f\in F} p(f) = 2 \crossnum(\mathcal S, \mathcal C)\le 2 \cdot k \cdot \# \mathcal S \le 2 (3g-1) k$.

Similarly let $n(f)$ denote the number of vertices of $\mathcal{C}$ in the face $f$. Then $\crossnum(\mathcal C)=\sum_{f\in F}n(f)$.
As $\mathcal C$ is in minimal position, the restriction of $\mathcal C$ to any face $f$ is also in minimal position. This implies that any two curves in $\mathcal{C}$ cross at most once in $f$, since otherwise they would form a bigon and thus would not be in minimal position. Thus $n(f)\le \frac{p(f) \cdot (p(f) - 1)}{2}$.
So
\[
\crossnum(\mathcal{C})
= \sum_{f\in F} n(f) \le \sum_{f\in F} \frac{p(f)^2}{2} - \frac{p(f)}{2} 
 \le \frac{1}{2}  \left(\sum_{f\in F} p(f) \right)^2
\le \frac{(2 (3g-1) k)^2}{2}.\qedhere
\]
\end{proof}

Combining Proposition~\ref{pr:filling-systoles} and Theorem~\ref{th:vertices-number-systolic} provides us with a first bound of $O(k^2g^2)$ on the number of vertices in a $k$-irreducible system of curves, and thus on the number of edges in a $k$-irreducible triangulation via the reduction in Section~\ref{S:irred}. In the next section we add another ingredient to strengthen this bound to $O(k^2 g)$.

\begin{remark}
Note that the family $\mathcal{C}$ in Theorem~\ref{th:vertices-number-systolic} is not assumed to be geodesically $k$-covered, and thus this theorem could be of independent interest. It is easy to see that Theorem~\ref{th:vertices-number-systolic} is tight up to a constant factor, by considering a genus $g$ orientable surface obtained by identifying opposite edges in a $2g$-gon, and a system of curves made of $k$ parallel curves between each pair of opposite edges.
\end{remark}

\section{Upper bounds for geodesically \boldmath\texorpdfstring{$k$}{k}-covered systems of curves}\label{S:linear}

In this section, we show: 

\kirreduciblesystem*

To prove this result, we first establish a preliminary lemma on configurations of arcs in a disk. An \define{arc} in a disk is a curve with endpoints on the boundary of the disk. For a system of arcs $\mathcal{A}$ in minimal position in a disk, we denote by $d$ the largest distance between a face of $\mathcal A$ and the boundary of the disk (where we measure lengths in terms of crossings with $\mathcal{A}$) and by $n$ the number of intersections between pairs of arcs of $\mathcal A$. The following inequality allows us to bound $n$ in terms of $d$ and $p$.

\begin{lemma}
\label{le:radius-perimeter-area-relation}
Let $\mathcal A$, $p$, $d$, $n$ be as above. Then $(6d+7)p\ge n$.
\end{lemma}

Our proof relies on a recent structural result of Hickingbotham, Illingworth, Mohar and Wood~\cite{hickingbotham-illingworth-mohar-wood-2024} and the fact that a graph of treewidth $k$ with $n$ vertices has less than $kn$ edges (see, e.g., Baste, Noy and Sau~\cite{baste2018number}).

\begin{proof}
Let $G$ be the intersection graph of $\mathcal A$: it has one vertex per arc of $\mathcal A$ and an edge between vertices corresponding to intersecting pairs of arcs. By definition, $G$ has $p$ vertices and $n$ edges.

Let $M_G$ be the map graph of $\mathcal A$ which has one vertex for each face of $\mathcal{A}$, including the outer face, and edges between vertices corresponding to faces that share at least one vertex. Let $\rad(M_G)$ be its radius, that is, the minimum radius of a ball that covers $M_G$. Then, by~\cite[Theorem~5]{hickingbotham-illingworth-mohar-wood-2024}, we have $\tw(G) \le 6 \cdot \rad(M_G)+7$, where $\operatorname{tw}(G)$ denotes the treewidth of $G$.

Let $v_\infty$ be the vertex of $M_G$ corresponding to the external face. Then $v_\infty$ is at distance at most $d$ from every face of $\mathcal A$ in the dual graph of $\mathcal A$. It follows from the definitions that the distances in the map graph are smaller than those in the dual graph, so $v_\infty$ is at distance at most $d$ in $M_G$ from every point of $M_G$. Thus $\rad(M_G)\le d$ and $\operatorname{tw}(G)\le 6d+7$.

Since a graph of treewidth $t$ and $v$ vertices has strictly less than $tv$ edges, $(6d+7)p\ge n$.
\end{proof}

We can now conclude the proof of Theorem~\ref{th:size-k-irreducible-system}.

\begin{proof}[Proof of Theorem~\ref{th:size-k-irreducible-system}]
The proof starts like the proof of Theorem~\ref{th:vertices-number-systolic}. By Proposition~\ref{pr:filling-systoles} and Lemma~\ref{le:genus-size-filling-family}, a geodesically $k$-covered system of curves is $2k$-saturating and thus there is a filling family $\mathcal{S}$ of at most $3g-1$ closed curves of length at most $2k$ with respect to $\mathcal{C}$. We think of $\mathcal{S}$ as a $4$-valent combinatorial surface with a family of faces $F$ and denote by $n(f)$ the number of vertices of $\mathcal{C}$ in a face $f$, and by $p(f)$ the perimeter of $f$, i.e., the number of intersections of $\mathcal{C}$ and the boundary of $f$. 

\begin{claim}
For any face $f$ of $F$, $n(f)\le (3k+7)\cdot p(f)$.
\end{claim}

\begin{proof}
\renewcommand\qedsymbol{$\lrcorner$}
In order to prove this claim we show that $d\leq \frac{k}{2}$. Since $\mathcal{C}$ is geodesically $k$-covered, for any point $R$ in $f$ that does not lie on an intersection of two arcs, there is a geodesic going through $R$ of length at most $k$. This geodesic must intersect the boundary of $f$, and thus $R$ is at distance at most $\frac{k}{2}$ from this boundary. Then we have $d\leq \frac{k}{2}$.

By Lemma~\ref{le:radius-perimeter-area-relation}, $(6d+7)p(f)\geq n(f)$. Thus $n(f)\le (3k+7)\cdot p(f)$.
\end{proof}

We conclude using $\sum_{f\in F}n(f)=\crossnum(\mathcal{C})$ and $\sum_{f\in F} p(f) = 2 \crossnum(\mathcal S, \mathcal C) \le 2 (3g-1) 2k$, which gives
\begin{align*}
\crossnum(\mathcal{C})= \sum_{f\in F} n(f) &\le \sum_{f\in F} (3k+7)\cdot p(f) \\
&\le (3k+7) \cdot 2 \cdot (3g-1) \cdot 2k \\
&\le 120k^2 \cdot g.\qedhere
\end{align*}

\end{proof}

\section{Non-orientable surfaces}\label{S:nonorientable}

In this section, we explain how to extend our result to the projective plane and show the limitation of our approach to the general case of non-orientable surfaces. We restate our result for the projective plane for convenience.

\nonorientable*

\begin{proof}[Proof of Theorem~\ref{t:projectiveplane}]
Let $T$ be a $k$-irreducible triangulation of the projective plane, i.e., $S$ has non-orientable genus $1$. As in the case of orientable surfaces, we consider an edge inclusion-wise minimal subgraph $G$ with the property that shortest non-contractible nooses on $G$ still have length $k$, and we denote by $\mathcal{C}$ the system of curves associated to $T$. Note that $\mathcal{C}$ is $2k$-irreducible, but Theorem~\ref{th:schrijver2} does not apply on non-orientable surfaces, and thus we do not know \emph{a priori} whether $\mathcal{C}$ is in minimal position. 

By assumption, there exists a geodesic $c$ of length $k$ going through every edge of $T$, we choose one arbitrarily. Note that $c$ must be a simple curve, otherwise there would be a shorter geodesic. On a projective plane, any non-contractible closed curve cuts the surface into a disk, so we can use $c$ as the analogue of the filling family we used for $k$-saturating systems of curves. We consider the system of arcs $\mathcal{A}$ induced by $\mathcal{C}$ on the disk $D:=S \setminus c$, and this system of arcs must be in minimal position. Indeed, the techniques of Schrijver~\cite[Proof of Main Lemma]{schrijver1991decomposition} do apply in the setting of a disk: if the system $\mathcal{A}$ was not in minimal position, there would be a monogon or a bigon in the disk $D$ (see for example Hass and Scott~\cite[Proof of Lemma~1.6]{hass1994shortening}), and this would allow for a smoothing that does not change the length of boundary to boundary paths in $D$. Since in the projective plane, any pair of non-contractible curves intersect, this smoothing would not change the length of geodesics, contradicting the $2k$-irreducibility of $\mathcal{C}$.

There remains to bound the size of a system of arcs in minimal position in a disk. We could use Lemma~\ref{le:radius-perimeter-area-relation} as before, but in this case there is a stronger and more direct bound. Here, the perimeter of the disk is $4k$ since the length of $c$ is $2k$ and gets doubled when cutting because $c$ is one-sided, and thus there are $2k$ simple arcs, crossing pairwise at most once. Actually, each pair of arcs has to cross exactly once, otherwise there would be a path of length less than $2k$ joining two antipodal points on the boundary of $D$ and thus the systole would be less than $2k$, contradicting $2k$-irreducibility. Such an arrangement of arcs is called an arrangement of pseudolines in the disk. This directly implies that there are exactly $\binom{2k}{2}=k(2k-1)$ crossings of $\mathcal{C}$ in $D$ and thus on $S$. 

We conclude by proceeding as in the proof of Theorem~\ref{th:size-k-irreducible-system} in Section~\ref{S:irred}. In order to do that, we claim that there are no monogons or bigons in $\mathcal{C}$: if there were such a monogon or bigon, after cutting it by $c$, we would observe in $D$ either a monogon, a bigon, two non-crossing parallel arcs, a bigon between an arc and $c$ or a pair of arcs crossing twice in $D$ (because $k>1$). All of these cases are incompatible with the fact that the arcs form an arrangement of pseudolines. We obtain that $|V(T)|\leq\frac{2}{3} \crossnum(\mathcal{C}) = \frac{2 k (2k-1)}{3} \leq \frac{4}{3} k^2$, and thus $|E(T)| \leq 4k^2$.
\end{proof}

In the case of the projective plane, this proof technique is very efficient and yields a tight bound on the number of crossings in a $k$-irreducible system of curves. However, there is an inherent loss when converting this bound in the setting of triangulations, and we obtain an upper bound of $2k(2k-1)/3$ on the number of vertices of a $k$-irreducible triangulation. This should be compared to a construction of Melzer~\cite[Section~6.2.1]{melzer2019k} with $k(k-1)+1$ vertices, which he conjectured to be optimal. Closing this gap remains an interesting open problem.

We now discuss the complication that appears for general non-orientable surfaces. The main difficulty is that Schrijver's Theorem~\ref{th:schrijver2} is not known in that setting. Given a $k$-irreducible triangulation $T$ on a non-orientable surface $S$, one could consider its lift $T'$ to the orientable cover $S'$ of $S$. The graph $T'$ is indeed a triangulation. The triangulation $T'$ is however generally not geodesically covered as shown in the following proposition.
\begin{figure}[!ht]
\begin{center}\includegraphics{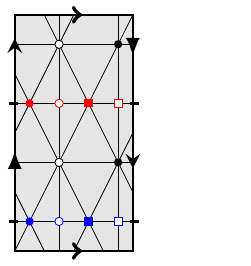}\hspace{1cm}\includegraphics{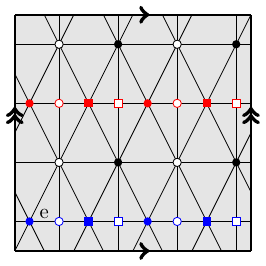}\end{center}
\caption{The 4-irreducible triangulation $T$ on the Klein bottle and its orientation cover $T'$ that appear in Proposition~\ref{prop:klein-bottle}.}
\label{fig:klein-bottle}
\end{figure}
\begin{proposition}\label{prop:klein-bottle}
Let $T$ and $T'$ be the triangulations of the Klein bottle and the torus from Figure~\ref{fig:klein-bottle}. Then $T$ is geodesically $4$-covered but its orientable cover $T'$ is not geodesically $k$-covered for any $k$.
\end{proposition}

\begin{proof}
Showing that $T$ is geodesically $4$-covered is tedious but straightforward. One just has to draw all non-contractible walks of length 4 on $T$ and verify that they cover all edges.

Proving that $T'$ is not geodesically $k$-covered for any $k$ is more delicate. Indeed, one has to exhibit an edge that is avoided by all geodesics. We claim that it is the case for the edge $e$ on Figure~\ref{fig:klein-bottle}. The proof can also be reduced to a finite tedious but straightforward check. Namely, given the neighborhood of $e$ shown on Figure~\ref{fig:unused-edge-neighborhood}, it is enough to check that all shortest paths from a boundary vertex to another avoid $e$.
\begin{figure}[!ht]
\begin{center}\includegraphics{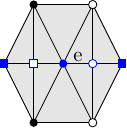}\end{center}
\caption{A neighborhood of the edge $e$ from the triangulation $T'$.}
\label{fig:unused-edge-neighborhood}
\end{figure}
\end{proof}

\subparagraph*{Acknowledgements.} We are grateful to anonymous reviewers for helpful remarks and suggestions.

\bibliographystyle{plainurl}
\bibliography{biblio.bib}

\end{document}